\documentclass[copyright]{eptcs}
\providecommand{\event}{HCVS/VPT 2022} %
\usepackage{breakurl}             %
\usepackage{underscore}           %

\usepackage{listings}
\usepackage{amsmath,amssymb}
\usepackage{amsthm}
\newtheorem{theorem}{Theorem}
\newtheorem{lemma}[theorem]{Lemma}
\usepackage{alltt}
\usepackage{stackrel} %

\usepackage{xcolor} %

\newcommand{\ie}{{\it i.e.}}
\newcommand{\eg}{{\it e.g.}}

\newcommand{\cf}{{\it cf.}}

\newcommand{\myem}[1]{\underline{#1}} %

\newif\ifdraft
\draftfalse  %

\ifdraft 
\usepackage{ulem}
\newcommand{\del}[1]{\textcolor{gray}{\sout{#1}}} %
\newcommand{\add}[1]{\textcolor{red}{#1}}         %
\newcommand{\addx}[1]{\textcolor{blue}{#1}}       %
\newcommand{\addz}[1]{\textcolor{red}{#1}}
\else
\newcommand{\del}[1]{}
\newcommand{\add}[1]{{#1}}
\newcommand{\addx}[1]{{#1}}
\newcommand{\addy}[1]{{#1}}
\newcommand{\addz}[1]{{#1}}
\fi

\title{Reversible Programming: \\ A Case Study of 
Two 
String-Matching Algorithms}

\author{Robert Gl\"uck %
\institute{DIKU, Department of Computer Science,\\
University of Copenhagen, Denmark 
\email{glueck@acm.org}}
\and
Tetsuo Yokoyama
\institute{Dept.\ of Electronics and Communication Technology,\\
Nanzan University, Japan
\email{tyokoyama@acm.org}}}

\def\titlerunning{Reversible Programming: A Case Study of Two String-Matching Algorithms}
\def\authorrunning{R.\ Gl\"uck \& T.\ Yokoyama}

\begin{document}
\maketitle

\begin{abstract}
String matching is a 
fundamental 
problem \add{in algorithm}.
This \add{study examines} the development and 
construction of two reversible string-matching algorithms:
a naive string-matching algorithm
and the Rabin--Karp algorithm. 
The algorithms \add{are used}
to introduce reversible computing concepts\add{, beginning} from basic reversible programming techniques to more advanced considerations about the injectivization of the polynomial hash\add{-}update function employed by the Rabin--Karp algorithm. 
The result\add{s }\add{are} two clean input-preserving reversible 
algorithms that require no additional space and have the same asymptotic time complexity as their classic irreversible originals.
This study aims to contribute to 
the body of reversible algorithms
and 
to the discipline of reversible programming.

\end{abstract}

\section{Introduction}
\label{sec:introduction}

Reversible computing is an unconventional computing paradigm in which all computations are forward and backward deterministic.
It complements existing mainstream programming paradigms that are forward deterministic, but usually backward nondeterministic, such as imperative and functional programming languages~\cite{GlYo22:RPLA}. Reversible computing is \add{required} when 
{the} deletion of information is considered harmful as in 
quantum-based computing and to overcome Landauer's physical limit~(for a summary see~\cite{DeVos:20,Krakovsky:21}).
\add{Additionally}, reversible computing 
is a sweet spot for 
studying non-standard semantics 
and %
program inversion, 
which concern fundamental questions \add{regarding}\del{about} program transformation~\cite{GlueckKlimov:94:LMMC}. 

The \del{aim of this
short}
contribution \add{of this study} is threefold:
\begin{itemize}
\item Introduce 
reversible computing concepts by 
a program development
in this
unconventional paradigm.
\item Explain
a new,
efficient reversible version of the Rabin--Karp algorithm for string matching.
\item Contribute to 
the advancement of a reversible programming discipline.
\end{itemize}
String matching is a 
fundamental algorithmic problem 
with a wide range of practical applications. The problem
is stated in a few lines (we follow established terminology~\cite{CLRS09:intro-alg}): 
\begin{quote}
Let $T[0..n-1]$ be a \emph{text} of length $n$ and $P[0..m-1]$ be a \emph{pattern} of length $m~(\leq n)$ where the elements of arrays $T$ and $P$ 
are characters drawn from a finite alphabet $\Sigma$ of size $d$. The arrays are also called \emph{strings} of characters.
A pattern $P$ occurs with a \emph{valid shift} $s$ in text $T$ if $T[s..s+m-1]=P[0..m-1]$. 
The \emph{string-matching problem} is to find all valid shifts of $P$ in~$T$.
\end{quote}

Today there 
are various string-matching algorithms all of which are defined in 
conventional languages.
In this  study, we develop
efficient
reversible 
string-matching algorithms, namely\add{,}
\add{a} naive algorithm and 
the more efficient Rabin--Karp algorithm~\cite{KaRa87}. \add{Although} the worst-case matching time of the latter is no better than \add{that} of the naive method, the Rabin--Karp algorithm is faster on average because it uses hash values for fast, approximate matches, and only in \add{the} case of a possible match\add{,} performs an exact comparison of the pattern and text at the current shift.
The use of a hash\add{-}update function that computes the next hash value from the current hash value (rolling hash)
makes the creation of an efficient and reversible Rabin--Karp algorithm more challenging than 
\add{that of} the naive algorithm.

We \add{develop} the reversible string-matching algorithms to explain 
reversible-computing concepts and how to solve the challenge posed by Rabin--Karp's hash-based method. \add{Specifically}, we use the standard reversible programming language\add{,} Janus\add{,} {with syntactic sugar} \add{to define}
the algorithms. Once the algorithms are written in a reversible language, they are guaranteed to be reversible.

More details about reversible computing can be found in the literature, {\eg}, \cite{YoAxGl:16:TCS}. 
This study \add{contributes to the existing literature}
\add{on} 
{clean}
reversible algorithms
{that do not rely on tracing}~({\eg}, \cite{AxYo15:sort,GlYo19:IPL}), including a reversible FFT~\cite{YoAG08a:CF}, Dijkstra's permutation encoder~\cite{YoAxGl:16:TCS}, {and language processors~\cite{YokoyamaGlueck:07:Janus,AxGl11LATA}}.

\section{A Reversible Naive String Matcher}
\label{sec:naivematcher}

We begin with the naive matcher to \add{demonstrate} how to \add{proceed with} reversible programming. The naive matcher develop\add{ed} in this section will later be \add{employed} in the reversible Rabin--Karp algorithm.
\add{Advanced} considerations, such as the injectivization of the hash\add{-}update function, are \add{discussed in} the next section about the Rabin--Karp algorithm.

Algorithms written in a reversible language cannot delete information, \addx{but can} \emph{reversibly update} information. No deletion of information means that programs written in a reversible imperative language cannot contain assignments that overwrite values, such as \mbox{\texttt{x := y + z}}, only {reversibly update} values, such as \mbox{\texttt{x += y}} (shorthand for \mbox{\texttt{x := x + y}}).\footnote{Any expression $e$ can be used in \mbox{\texttt{x += }$e$}.
\add{Generally, to} be reversible,
\texttt{x} must not occur in $e$ ({\eg},
\mbox{\texttt{x += -x}} is not reversible).}
\add{Consequently,} reversible languages can only be {as} computationally powerful as \emph{reversible Turing machines} (RTMs)~\cite{Bennett:73}, which \add{exactly} compute 
the \emph{computable injective {functions}}. This injectiveness constraint makes reversible Turing machines strictly less powerful \add{compared to} their traditional counterparts that are universal.
This \add{is} a \add{significant} limitation of reversible computing\add{;} \add{however,} all non-injective functions can be \emph{embedded} in injective \add{functions}.

\add{There} are two approaches 
\add{for} implementing a function in 
reversible language. The first approach is to \add{begin} from an implementation written in a conventional (irreversible) language and \emph{reversibilize} it into a program written in a reversible programming language, {\eg}, by recording the information otherwise lost (often called \emph{garbage}). The second
and \add{the} preferred approach is to \add{change} (\emph{injectivize}) the problem specification 
into an injective function, which 
can be 
\add{directly} implemented in a reversible language without functional change\add{s}.

The string-matching problem has 
an injective specification
\addx{although} it is usually \add{considered} a non-injective problem that, given text $T$ and pattern $P$, computes all valid shifts:
\begin{eqnarray}\label{eq:noninj}
\textit{match}(T,P) &=& \textit{valid-shifts}\,.
\end{eqnarray}
This function specifies that $T$ and $P$ are consumed (deleted) by $\textit{match}$ and \add{are} replaced by $\textit{valid-shifts}$.
Considering the problem from a reversible-computing perspective, we notice that we usually do not delete $T$ and $P$, but preserve them. 
This means that the
string-matching problem has an injective specification:
\begin{eqnarray}\label{eq:inj}
\textit{match}(T,P) &=& (T,P,\textit{valid-shifts})\,.
\end{eqnarray}
Because of the injective specification,
a 
faithful 
reversible implementation of the string-matching problem 
exists.
This specification is an \emph{input-preserving injectivization} ($T$ and $P$ are preserved).

A naive string-matching algorithm compares %
$P$ \add{to} 
$T$ at all 
shifts in $T$ 
from left to right. 
When a mismatch is found at a shift $s$, \add{the} matching continues at the next shift $s+1$. The worst-case matching time \add{for} this (irreversible) naive matching method is $\Theta((n-m+1)m)$.

First, we consider certain standard technical details. We consider characters as integers such that alphabet $\Sigma$ of size $d$ is set $\{0,\ldots,d-1\}$. 
\add{Let} $T$ and $P$ be integer arrays of length $n+1$ and $m+1$, respectively, with 
terminating values $T[n] \neq P[m]$, where $P[m]\not\in \Sigma$. 
\add{Thus,} the end of $p$ is always signaled by a mismatch. 
All valid shifts $s_i$ found during \add{the} search are pushed on a result stack.
\add{Thus,}
\textit{valid-shifts} in Eqs.~(\ref{eq:noninj}, \ref{eq:inj}) is a stack,
\add{which consists} of zero or more unique indices 
of $T$.

\addx{At first glance,} the reversible naive string matcher in Fig.~\ref{fig:naive} looks like a C-like program. \addx{At the second look, we notice that} {the program} uses no destructive assignments, such as \texttt{:=}, only the reversible updates \texttt{+=} and \texttt{-=} that add to resp.\ subtract from a variable the value of an expression,
and 
the conditional at lines~\ref{lst:ns:if} to~\ref{lst:ns:fi} not only has an entry test at \texttt{if}, but also an exit test at~\texttt{fi}, which is the point
where the control flow joins after executing one of \addx{the} two branches. 

Reversible languages \add{comprise} elementary steps that perform injective transformations of the computation state, that is each step performs a forward and backward deterministic transition. 
Because the operations are reversible on the microscopic level, the macroscopic operation of a program written in a reversible language is perfectly reversible. The composition of injective functions is also an injective function, thus reversible programs implement computable injective functions. This principle is the same for all reversible languages including the transition function of a reversible Turing machine~\cite{Bennett:73}, a time-symmetric  machine~\cite{Nakano:22:timesymmetric}, and extensions 
that operate on quantum data (for quantum circuits, {\eg}~\cite{OOCG20}).

\paragraph{A Reversible Matcher}

The reversible naive string matcher\add{, shown} in Fig.~\ref{fig:naive} consists of three  procedures. 
The main procedure \texttt{naivesearch} is called with a text \texttt{T}, a pattern \texttt{P}, and an initially empty stack \texttt{R} as input. When it returns, all valid shifts are stored \add{in} \texttt{R}\addx{,} and \texttt{T} and \texttt{P} are unchanged (all three arguments are pass-by-reference).
The procedure \addx{tries} %
all \add{the} possible shifts 
from left to right
by incrementing \texttt{s} from \texttt{0} to \texttt{n-m} 
and calling procedure \texttt{match} 
in the for-loop \texttt{iterate} in lines~\ref{lst:ns:iterate:begin} to~\ref{lst:ns:iterate:end}. 
As a shorthand, we write \texttt{m} and \texttt{n} for the size of the \add{pattern and }text\add{, respectively}.

Procedure \texttt{match} begins \add{with} calling 
procedure \texttt{compare} to match \texttt{P} \add{to} \texttt{T} \add{beginning} at shift \texttt{s} with the initial index \texttt{i} = \texttt{0}.
If \texttt{compare} returns with \texttt{i} = \texttt{m} (end of \texttt{P} \add{is} reached), the match \add{succeeds,}  and \texttt{s} is a valid shift; otherwise, the match \add{fails} (end of \texttt{P} \add{is} not reached).
The then-branch pushes the valid shift \texttt{s} \add{to} \texttt{R} and resets \texttt{i} to zero. Line~\ref{lst:ns:top} is \add{required} to restore the last value of \texttt{s} from the top of \texttt{R} because the last \texttt{push} \emph{moved} that value to \texttt{R} and thereby zero-cleared \texttt{s}. (This point \add{is} explained below.) %
In the else-branch, after 
the match \add{fails} at \texttt{i < m}, the computation of \texttt{compare} 
is undone 
by uncalling \texttt{compare} to reset \texttt{i} to its initial value \texttt{0}. (This \add{is} also 
explained below.) 
A local scope for \texttt{i} is opened and closed at the begin\add{ning} and end of
\texttt{match} 
\add{using} a \texttt{local}--\texttt{delocal} declaration. \add{Furthermore, this} scope declaration assert\add{s} the initial and final values of local variable \texttt{i} (in both cases \texttt{i} = \texttt{0}).

Procedure \texttt{compare} compares \texttt{P} \add{with} \texttt{T} at \texttt{s}\add{, beginning} with index \texttt{i} = \texttt{0}.
The loop begins with an entry test at line~\ref{lst:ns:from}, which asserts that initially \texttt{i} = \texttt{0}, and ends at line~\ref{lst:ns:until} when 
\texttt{T[s+i]} $\neq$ \texttt{P[i]}.
\add{This} loop always terminates 
because by convention the terminating value \texttt{P[m]} is not in \texttt{T}. 

\begin{figure}
\begin{lstlisting}[texcl=true]
procedure compare(int T[], P[], s, i)
  from i = 0 loop       (*@\label{lst:ns:from}@*)// index assertion
    i += 1              (*@\label{lst:ns:body}@*)// character-by-character comparison
  until T[s+i] != P[i]  (*@\label{lst:ns:until}@*)// loop terminates when a mismatch occurs
    
procedure match(int T[], P[], s, stack R)
  local int i = 0           (*@\label{lst:ns:local}@*)
  call compare(T, P, s, i)  (*@\label{lst:ns:call:cmp}@*)
  if i = m then  (*@\label{lst:ns:if}@*)// match succeeded
    push(s, R)   (*@\label{lst:ns:push}@*)// push s to stack R of valid shifts w/ clearing s
    s += top(R)  (*@\label{lst:ns:top}@*)// restore the value of s
    i -= m       (*@\label{lst:ns:i-=m}@*)// clear i
  else           (*@\label{lst:ns:else}@*)// match failed
    uncall compare(T, P, s, i) (*@\label{lst:ns:uncall}@*)// clear i
  fi s = top(R)  (*@\label{lst:ns:fi}@*)// current shift s is valid
  delocal int i = 0         (*@\label{lst:ns:delocal}@*)

procedure naivesearch(int T[], P[], stack R)
  iterate int s = 0 to n-m  (*@\label{lst:ns:iterate:begin}@*)// slide over text
    call match(T, P, s, R)  (*@\label{lst:ns:call:match}@*)// match at current shift s
  end                       (*@\label{lst:ns:iterate:end}@*)
\end{lstlisting}
\caption{\add{Reversible} naive string-matching algorithm.}
\label{fig:naive}
\end{figure}

\paragraph{Reversible Programming}

Similar to \add{the} other language paradigms, 
reversible computing has its own programming \addx{methodology}. We summarize the programming techniques relevant %
{to} the programs in this \add{study} and exemplify them with examples from the programs. \add{This is related} to three important reversible programming
themes: control flow, reversible updates, and data structures.

\begin{description}

\item{Control flow:}
Join points in the control flow of a program require assertions to make them \emph{backward deterministic}. In reversible languages, each {join} point is associated with a predicate that provides an assertion \add{regarding} the incoming \add{computational} states.
This \add{suggests that} we \add{must identify} a predicate that is true when coming from a then-branch and false when coming from an else-branch. For a loop\add{,} we \add{must identify} a predicate that is initially true and false after each 
{iteration}.
These assertions \add{regarding} the incoming control flow (`\emph{come from}') are 
evaluated at runtime\add{, similar to} the tests that dispatch the outgoing control flow (`\emph{go to}'). If a predicate does not have the expected truth value, the control-flow operator is undefined, and therefore the entire program. 

Examples are \texttt{fi}-predicates in \emph{reversible conditional} (\texttt{if-fi}) and \texttt{from}-predicates in a \emph{reversible while-loop} (\texttt{from-until}).  
Sometimes\add{,} these assertions are easy to find, such as the entry test 
in line~\ref{lst:ns:from} of the increment loop in Fig.~\ref{fig:naive}, which checks that the loop \add{begins} from \mbox{\texttt{i} $=$ \texttt{0}} and \mbox{\texttt{i} $\neq$ \texttt{0}} after the first iteration.
The exit test \texttt{s = top(R)} in line~\ref{lst:ns:fi} of the conditional uses the fact that the shifts \add{in} stack \texttt{R} are unique. \add{Thus}, whenever a match \add{fails, indicating that} \texttt{s} is not pushed to \texttt{R}, 
the current 
shift \texttt{s} and the last shift \texttt{top(R)} 
differ. An excerpt \add{from} the program highlights these two cases:
\begin{alltt}
  \myem{from i = 0} loop          if i = m then push(s,R)\textrm{ ...}
    i += 1                   else \textrm{... no push ...}
  until T[s+i] != P[i]     \myem{fi s = top(R)}
\end{alltt}
However, these
assertions are not always easy to find and may require a restructuring of the program. 
Only a few conventional control-flow operators are reversible and do not require additional assertions, such as for-loops that iterate for a fixed number of times, {\eg}, \texttt{iterate} in lines~\ref{lst:ns:iterate:begin}--\ref{lst:ns:iterate:end}. 

\item{Reversible updates}: Data can only be \add{reversibly }updated.
The usual computational resource\add{s} \add{for}
deleting 
data 
in one way or another 
\add{are} not available
({\eg}, forgetting local variables upon procedure return). 
We \add{present} several update techniques 
used in our programs starting from a straightforward initialization of a zero-cleared variable to the uncalling of a procedure to reset values.
Readers interested in reversible updates defined in a more general form should refer to \cite{YoAG08a:CF}.

(i) 
Copying \& zero-clearing. 
If variable \texttt{i} is known to be zero, it can be set to a value, {\eg}, by addition. For example, \mbox{\texttt{i += m}} has the effect of reversibly copying the value of \texttt{m} to \texttt{i}. Similarly, if we know that \texttt{i} has the same value as \texttt{m}, that is \mbox{\texttt{i} $=$ \texttt{m}}, we can zero-clear \texttt{i} \add{using} \mbox{\texttt{i -= m}}. 
\add{However, the} relation\add{ship} \add{between} two variable values is not always known.
\add{Additionally}, when it becomes known \add{owing} to an equality test in a conditional, we can exploit this knowledge in the then-branch to zero-clear the variable. This is used in line~\ref{lst:ns:i-=m} to reversibly reset \texttt{i} to zero:
\begin{alltt}
  if \myem{i = m} then \textrm{...} 
    \myem{i -= m} 
  else \textrm{...}  
  fi \textrm{...}  
\end{alltt}
These techniques are \add{indirectly} used in a local--delocal declaration\add{,} where the local variable is initialized and cleared at the begin\add{ning} and end 
\addx{of its scope} 
\add{using} an equality test (here, however, \addx{just} a simple \mbox{\texttt{x = 0}} in lines~\ref{lst:ns:local} and~\ref{lst:ns:delocal}).
\addx{In general}, the declaration of a local variable\add{,} \texttt{i}\add{,} has the following form, where \texttt{i} is initially set to the value of $e$ and in the end must have a value equal to the value of~$e'$:
\begin{alltt}
  local int \myem{i = \ensuremath{e}} \textrm{...} delocal int \myem{i = \ensuremath{e'}}
\end{alltt}

(ii) Compute-uncompute.
Reversible programs are forward and backward deterministic\add{; thus,} they can run 
\addx{efficiently} 
in both directions.
Many reversible languages not only provide 
access to their standard 
semantics, {\eg}, \add{using} a procedure call, but also to 
their inverse (backward) semantics, {\eg}, \add{using} a procedure uncall. An uncall of a procedure is as efficient as a call because a procedure is forward and backward deterministic.
We \add{employ}
this property
to reset 
index \texttt{i} after an unsuccessful match, which can occur at any
position \mbox{\texttt{i} $<$ \texttt{m}} in a pattern.
We
cannot determine the subtrahend
to zero-clear \texttt{i} (and we cannot use the irreversible \mbox{\texttt{i -= i}}). Instead we undo the computation of \texttt{i} by an uncall in the else-branch. This resets \texttt{i} to its initial value \texttt{0}.
\add{By combining} the techniques  
seen so far, 
we ensure that \texttt{i} is zero-cleared after the \texttt{if-fi}. 
In the then-branch\add{,} we \add{use} 
the equality \texttt{i = m}; 
in the else-branch we undo the computation: 
\begin{alltt}
  local int i = 0
    \myem{call compare(\textrm{...},i)}
    if i = m  then \textrm{...} 
      i -= m
    else \myem{uncall compare(\textrm{...},i)}
    fi \textrm{...}
  delocal int i = 0
\end{alltt}
The compute-uncompute method 
goes back to the first 
RTMs~\cite{Bennett:73}, where a machine is textually composed with its inverse machine to restore the original computation state (which doubles the size of the entire machine). The call--uncall method above shares the text of a procedure (here, \texttt{compare}). It just invokes \add{the} standard resp.\ inverse computation \add{of the procedure}.
We could have used an uncall in the then-branch.
\add{Instead,} we exploit
the knowledge about \texttt{i = m} from the entry test of the \texttt{if-fi}
to \addz{zero-clear} 
\texttt{i} \add{using} \texttt{i -= m}, which takes constant time, whereas the uncall \add{requires} time proportional to the length of~\texttt{P}. \addz{The conditional takes advantage of both techniques.}

Bennett used \emph{program inversion} to obtain \add{an} inverse \add{RTM}, whereas the \add{aforementioned} method %
uses \emph{inverse computation}. We could have used program inversion to invert the procedure \texttt{compare} into the inverse procedure \texttt{compare}$^{-1}$ and invoked the latter \add{using} \texttt{call compare}$^{-1}$ to reset \texttt{i}.
\addz{Both methods, calling the inverse procedure \texttt{compare}$^{-1}$ and inverse computation of \texttt{compare},}
\begin{alltt}
  call compare\(\sp{-1}\)(\textrm{...},i)  \textrm{and}  uncall compare(\textrm{...},i)\textrm{,}
\end{alltt}
\addz{are functionally equivalent.}
Because RTMs cannot access their inverse semantics, \addz{{\eg}, by an uncall}, program inversion is the only choice to build \addz{RTMs} that restore the input from their output,
whereas in a reversible language typically both choices are available. We refer to 
them 
collectively as the \emph{compute--uncompute} 
programming method. It is 
used in many forms \add{at} all levels of a reversible computing system from reversible circuits{~(\eg, \cite{ThAxGl:11:RC})} 
to high-level languages{~(\eg,~\cite{Mogensen:22})}.

\item{Data structures:} The data structures in reversible languages are the same as in conventional languages, such as arrays, stacks and lists, only the update operations on the data structures must be reversible. In the case of a stack, the operations push and pop can be defined \add{as} inverse to each other by letting them swap in and out the value on top of the stack, which means that $\mathit{pop} = \mathit{push}^{-1}$~\cite{YokoyamaGlueck:07:Janus}: 
\begin{equation}
  (v, v_n\,...\,v_1)
  \begin{array}{c}
    \stackrel{\mathit{push}}{\longrightarrow} \\[-1ex]
    \stackrel[\mathit{pop}]{}{\longleftarrow}
  \end{array}
  (0, v\;v_n\,...\,v_1)
\end{equation}
This definition of a push has the unfamiliar property that \texttt{push(s,R)} \add{in} line~\ref{lst:ns:push} moves the value $v$ of \texttt{s} to the top of the stack \texttt{R} and zero-clears \texttt{s}. 
Because we need the value that we \add{have} pushed to continue the search,
the value is copied back to 
\texttt{s} from the top of \texttt{R} \add{using} \mbox{\texttt{s += top(R)}} in line~\ref{lst:ns:top}.

We can now complete 
the body of procedure \texttt{match} in lines~\ref{lst:ns:local}--\ref{lst:ns:delocal} 
by adding the two statements to the then-branch and the exit test 
that we discussed above
to \texttt{fi}:
\begin{alltt}
  local int i = 0
    call compare(\textrm{...},i)
    if i = m then
      \myem{push(s,R)}
      \myem{s += top(R)}
      i -= m
    else uncall compare(\textrm{...},i)
    fi s = top(R)
  delocal int i = 0
\end{alltt}
We remark
that abstract data types and object-oriented features can be used in reversible languages provided \add{that} their update operations and methods are reversible. 
Ideally, they are designed such that call and uncall \add{can be used.} \addx{When operators are inverse to each other, only one of them needs to be implemented.}
The idea of \emph{code sharing} 
by running code backward
can be traced back to the 60s~\cite{ReillyFederighi:65}.

\end{description}
We have presented all 
reversible-programming techniques used in 
the procedures \texttt{naivesearch}, \texttt{match}, and \texttt{compare}. This completes the review of the reversible naive string-matching algorithm \add{shown} in Fig.~\ref{fig:naive}.

\section{{A} Reversible Rabin--Karp Algorithm}
\label{sec:rabinkarp}

The Rabin--Karp 
algorithm~\cite{KaRa87} replaces exact matches \add{with} approximate matches, which are inexact but \addx{fast}, and performs exact matches only if a successful match is possible. This makes the algorithm conceptually easy and fast in practice.
{This study 
\addx{refers} to 
the version presented by Cormen et al.~\cite{CLRS09:intro-alg}.}

The algorithm extends the naive string-matching algorithm by performing at each shift $s$\add{,} an approximate match by comparing the \emph{hash values} of $P$ and $T$ at $s$. An exact match 
(by procedure \texttt{match}) 
\addy{can only succeed}
if the two hash values are identical; otherwise, an exact match is impossible. In either case, the next hash value at $s+1$ can be computed in constant time from the current hash value at $s$ (rolling hash) \add{using} a \emph{hash\add{-}update function} $\phi_s$. 
Hash values typically fit into single words that can be compared 
in constant time. The initial hash values of $P$ and 
$T$ at shift 0 are computed at the \add{beginning} of the algorithm \add{using} a \emph{hash function} (pre-processing). The \add{subsequent} hash values are then computed \add{using} the hash\add{-}update function.

The key to an efficient clean reversible Rabin--Karp matcher \add{is a} reversible constant-time 
\addx{calculation} of
\addx{the} rolling hash values.
We have explained the reversible naive string\add{-}matching algorithm in the previous section\add{,} and we can reuse 
procedures \texttt{match} and \texttt{compare} from Fig.~\ref{fig:naive} \add{for} the reversible Rabin--Karp algorithm. In this section, we focus on the injectivization and implementation of the hash functions \add{that}
show the considerations \add{for} the development of a more advanced reversible algorithm.

A preliminary version of the reversible Rabin--Karp program {has} appeared {as a} technical report~\cite{TaHY22:eng}. 
\add{The reversible Rabin--Karp program in this study becomes more concise and modular because of the use of macros and iterate loops.}

\paragraph{Hash Function}
The Rabin--Karp algorithm requires a \addx{pre-process that calculates}\del{ as 
a pre-process
the calculation of}
the hash values 
of the given pattern\add{,} $P[0..m-1]$\add{,} and of the initial substring\add{,} $T[0..m-1]$\add{,} of the given text $T$. The \addx{initial} substring has the same length $m$ as $P$. 
Recall that $P$ and $T$ are two integer arrays over the \addx{non-empty} alphabet of $d$ integers \mbox{$\{0,\ldots,d-1\}$}. 

Let $p$ denote the hash value of 
$P$
obtained \add{using a} polynomial hash function with modulus~$q$:
\begin{eqnarray}\label{eq:p}
p &=& (P[0]d^{m-1}+P[1]d^{m-2}+ \cdots + P[m-1]) \bmod q\,. 
\end{eqnarray}
Similarly, 
let $t_s$ denote the hash value of the substring $T[s..s+m-1]$ 
of length $m$ 
at shift $s$:
\begin{eqnarray}\label{eq:t}
t_s &=& (T[s]d^{m-1}+T[s+1]d^{m-2}+ \cdots + T[s+m-1]) \bmod q\,.
\end{eqnarray}
The polynomials can be computed 
in $\Theta(m)$ using Horner's rule.
Preferably, %
modulus $q$ \add{should be} as large a prime as possible such that $dq$ fits into a single word: \add{thus,} 
all \addx{modulo} operations are single-precision arithmetic. 

\add{The following properties of the two hash values are important.}
If $t_s\neq p$, 
$T[s..s+m-1] \neq P[0..m-1]$: 
thus, shift $s$  
{cannot be} valid. If $t_s = p$, 
it is \emph{possible} that $T[s..s+m-1] = P[0..m-1]$: 
thus, an exact 
match is \add{required} to determine \add{if} \addx{shift} $s$ is valid.

\addx{Whereas} \addx{hash} value $p$ of $P$
\add{is} the same during the matching, %
hash value $t_s$ of $T$
\add{must} be calculated for each shift $s$. \addx{In order to efficiently calculate hash values $t_s$ for $s>0$,} we 
calculate the hash value $t_{s+1}$ at the \add{subsequent} shift $s+1$ from the hash value $t_s$ at the current shift $s$\add{, using}
a recurrence function.
We use this 
function to reversibly update the hash values \addx{in constant time}. 
\addy{For a reversible implementation, conditions are first determined under which the function is injective. Then the function is rewritten into a composition of modular arithmetic operators, each of which is embedded in a reversible update.}

\paragraph{Hash\add{-}Update Function}
We can compute $t_{s+1}$ from $t_s$ 
for $0\leq s \leq n-m$ because
\begin{eqnarray}\label{eq:update-t}
t_{s+1} &=& \phi_s(t_s)
\end{eqnarray}
where the recurrence function {$\phi_s$} is defined \add{as}
\begin{eqnarray}\label{eq:phi}
\phi_s(x) &=& (d(x-T[s]d^{m-1})+T[s+m]) \bmod q\,.
\end{eqnarray}
The recurrence function 
calculates the \add{subsequent} hash value from the current hash value $x$ by canceling the old highest-order radix-$d$ digit $T[s]$ \add{via} subtraction, shifting the value \add{via} multiplication with $d$, adding the new lowest-order radix-$d$ digit $T[s+m]$\addx{, and obtaining its remainder when divided by $q$}. \addx{We can compute 
$\phi_s(t_s)$ in constant time 
if factor $d^{m-1}$ is precomputed.}
\addx{For reversible computing it is
problematic that}
$\phi_s$ is \add{generally} \emph{not injective} because the modulo operation is not injective \addx{in its first argument} for arbitrary integers.
\add{Determining the} conditions under which a function becomes injective by exploiting its properties and 
specific application context is an important step in the development of a clean reversible algorithm. 

We 
exploit certain properties in the domain of arithmetic operations in $\phi_s$. \addx{The congruence of two integers $x$ and $y$ modulo $q$, $x \equiv y \pmod q$, is compatible with addition, subtraction, and multiplication. But unlike these operations, division cannot always be performed. To show the injectiveness of $\phi_s$ in the following lemma requires that $d$ and $q$ 
be \emph{coprime} \addx{integers}, which means their greatest common divisor is \addx{1}.
For example, when $d$ and $q$ are not coprime, {\eg}, $d=2$ and $q=6$ then
$2\cdot 1 \equiv 2 \cdot 4 \pmod 6$, 
we cannot divide this congruence by 2 because $1 \not\equiv 4 \pmod 6$. \addx{Whereas} if they are coprime, {\eg},
$d=2$ and $q=3$, then
$2\cdot 1 \equiv 2 \cdot 4 \pmod 3$ and we can divide this by 2 to deduce $1 \equiv 4 \pmod 3$.}

\addx{We employ this constraint on the operands 
in Lemma~\ref{lem:injphi},
\add{which} \addx{establishes} that $\phi_s$ is injective. For the cancelation of common terms in congruences modulo $q$, we use the following two lemmas in the proof.}

\begin{lemma}[\eg{}~{\cite[Prop.\,13.3]{Lieb15}}] %
\label{lem:addition}
Suppose $x, y \in \mathbb{Z}$ and $x \equiv y \pmod q$. If $c \in \mathbb{Z}$ then
\begin{eqnarray}
x + c \equiv y + c \pmod q &\Longrightarrow& x \equiv y \pmod q \,.
\end{eqnarray}
\end{lemma}

\begin{lemma}[\eg{}~{\cite[Prop.\,13.5]{Lieb15}}] %
\label{lem:division}
Suppose $d$ and $q$ are coprime integers. If $x, y \in \mathbb{Z}$ then
\begin{eqnarray}
dx \equiv dy \pmod q &\Longrightarrow& x \equiv y \pmod q\,.
\end{eqnarray}
\end{lemma}
\begin{proof}
\add{We} assume that $dx \equiv dy \pmod q$. \add{Therefore,} $d(x-y)$ is a multiple of $q$.
Because $d$ and $q$ are coprime, $(x-y)$ is a multiple of $q$, \ie{} $x \equiv y \pmod q$.
\end{proof}%

\begin{lemma}
\label{lem:injphi}
Provided that $0\leq x < q$, $0< d < q$, and $d$ and $q$ are coprime, %
recurrence function $\phi_s$ is injective, where
\begin{eqnarray}\label{eq:phis}
\phi_s(x) &=& (d(x-T[s]d^{m-1})+T[s+m]) \bmod q\,.
\end{eqnarray}
\end{lemma}
\begin{proof}
We show that for any $x_1$ and $x_2$, whenever $\phi_s(x_1)=\phi_s(x_2)$, we have $x_1=x_2$.
\begin{align*}
 &~ \phi_s(x_1) ~~=~~ \phi_s(x_2)\\
\Longrightarrow~~
 &~ d(x_1-T[s]d^{m-1})+T[s+m] ~~\equiv~~ d(x_2-T[s]d^{m-1})+T[s+m] \pmod q\\
\Longrightarrow~~
 &~ d(x_1-T[s]d^{m-1}) ~~\equiv~~ d(x_2-T[s]d^{m-1}) \pmod q \qquad \because \text{Lemma \ref{lem:addition}}\\
\Longrightarrow~~
 &~ x_1-T[s]d^{m-1} ~~\equiv~~ x_2-T[s]d^{m-1} \pmod q \qquad\because \text{Lemma \ref{lem:division} and \addx{$d$ and $q$ are coprime integers}} %
\\
\Longrightarrow~~
 &~ x_1 ~~\equiv~~ x_2 \pmod q \qquad \because \text{Lemma \ref{lem:addition}}\\
\Longrightarrow~~
 &~ x_1 ~~=~~ x_2 \qquad\because 0\leq x_1,x_2 < q
\end{align*}
\end{proof}

The condition that $d$ and $q$ are coprime is not a restriction in practice because $q$ is usually \add{selected} to be a large prime number. Thus, if $q$ is prime, any alphabet size $0<d<q$ can be used.

\paragraph{Injective Modular Arithmetic}%

For 
efficient calculation of a hash value\add{,} it is preferable to perform modular arithmetic at each elementary arithmetic operation instead of first calculating a large value that may exceed the largest representable integer.
\add{Therefore}, 
we 
define the injective recurrence function~$\phi_s$ using elementary \addx{modular arithmetic operators}, each of which is \add{sufficiently} simple 
to be easily implemented in a reversible language (\eg, supported as built-in operators or by reversible hardware).
We rewrite $\phi_s(x)$ in Eq.~(\ref{eq:phis}) into a composition of \emph{modular arithmetic operators} $+_q$, $-_q$, and $\cdot_q$ \addx{using the distributivity of the modulo operation} as shown in Eq.~(\ref{eq:phismod}). \add{Similarly}, the hash functions in \add{Eqs.~(\ref{eq:p}) and (\ref{eq:t})} \add{can be} rewritten \add{as} \add{Eqs.~(\ref{eq:hashp}) \add{and} (\ref{eq:hasht})}. The same transformation is \add{required} for 
implementation in 
conventional \add{languages}. 
\add{Additionally}, we 
inspect the injectivity of each operator \add{such} that a reversible implementation can be provided.

The injective function\add{,} $\phi_s$\add{,}
\add{comprises the following} 
operators: 
\begin{eqnarray}
\phi_s(x) &=& d\cdot_q(x-_qT[s]\cdot_qh)+_qT[s+m]
\label{eq:phismod}
\end{eqnarray}
where 
factor $h=d^{m-1} \pmod q$ \add{is}
precomputed by
\begin{eqnarray}
h &=& \underbrace{d\cdot_q \cdots \cdot_q d}_{m-1}\,.
\label{eq:h}
\end{eqnarray}

Hash value $p$ of a given pattern \addx{$P$} and 
initial hash value $t_0$ of a given text \addx{$T$} can be obtained 
using Horner's rule {defined} in $\psi$:
\begin{eqnarray} %
p &=& \psi(P[0..m-1])
\label{eq:hashp} \\
t_0 &=& \psi(T[0..m-1])
\label{eq:hasht}
\end{eqnarray}
where 
\begin{eqnarray}
\psi(X[i..j]) &=& X[j] +_q d\cdot_q(X[j-1] +_q d\cdot_q(\cdots +_q d\cdot_q(X[i+1] +_q d\cdot_q X[i])\cdots))\,.
\label{eq:psi}
\end{eqnarray}

The identity between $\psi$ applied 
to the final substring\add{,} $T[n-m..n-1]$\add{,} at shift $n-m$
and the 
hash update $\phi_{n-m-1}$ \add{is} used to zero-clear the final hash value\add{,}
$t_{n-m}$\add{,} in the reversible program:
\begin{eqnarray} 
t_{n-m} &=& \psi(T[n-m..n-1]) \qquad\because\mbox{Eq.~(\ref{eq:psi})}\\
&=& \phi_{n-m-1}(t_{n-m-1})  \qquad\because\mbox{Eq.~(\ref{eq:update-t})} \,.
\label{eq:hashtfinal}
\end{eqnarray}%

The problem is that 
{modular arithmetic operations}
are \add{generally non-injective}.
However, under certain conditions, they are injective in \addx{one} of their arguments. \addx{For example, they are injective in their \addy{first} arguments:}%
\footnote{{A partial function\add{,} 
$f\!: X\times Y\times \cdots \times Y \rightharpoonup Z$\add{,} 
is \emph{injective in its first argument} \addx{iff} $\forall x_1, x_2 \in X$, $\forall y_1,\ldots,y_n \in Y$: if $f(x_1,y_1,\ldots,y_n)$ and $f(x_2,y_1,\ldots,y_n)$ are defined\add{,} $f(x_1,y_1,\ldots,y_n)=f(x_2,y_1,\ldots,y_n) \implies x_1 = x_2$.
\add{Similarly}, for other arguments\addx{~\cite{YoAG08a:CF}}.}}
\begin{itemize}
  \item Addition
  $x+_qy$ and subtraction 
$ x-_qy$
are injective 
in their \addy{first} arguments 
if 
$0\leq x, y < q$.
  \item Multiplication 
$x\cdot_qd$ is injective
in its \addy{first} argument 
if $0\leq x < q$, $1\leq d < q$, and $d$ and $q$ are coprime. %
\end{itemize}
Note that $x\cdot_qd$ is not injective in its \addy{first} argument\add{,} unless $d$ and $q$ are coprime.
Thus, \add{the relationship between} $d$ and $q$ is a necessary condition. 
Analogously, \addy{the same holds for the second} arguments of $+_q$, $-_q$, and $\cdot_q$.
\addx{Recall that the} composition of injective functions is an injective function. Thus, the injectiveness of
operators $x+_qy$, $x-_qy$, and $x\cdot_qy$ under the stated conditions \add{demonstrates} the injectiveness of 
\addy{Eq.~(\ref{eq:phismod})} under corresponding conditions.

\paragraph{Implementation of Modular Arithmetic}
A ternary
function that is injective in its first argument $f(x,y,z)$ can be embedded in \addx{the} \emph{reversible update} $g(x,y,z)=(f(x,y,z),y,z)$.
Arguments $y$ and $z$ are \addx{part of} the result of $g$. \add{Thus,} $g$ is injective.
Such reversible updates for 
$x+_qy$ and $x\cdot_qy$
can be implemented 
in Janus.
We
write 
\begin{itemize}
\item \texttt{x +=\(\sb{\texttt{q}}\) y} ~for~ $g_1(x,y,q)=(x+_qy,y,q)$, ~and
\item \texttt{x *=\(\sb{\texttt{q}}\) y} ~for~ $g_2(x,y,q)=(x\cdot_qy,y,q)$.
\end{itemize}

In the implementation of these operators in a reversible language, \texttt{x}, \texttt{y}, and \texttt{q} are integers that \addy{cannot be larger than the} largest representable integer in that language. 
Otherwise, the same restrictions as \add{those} for 
mathematical modular arithmetic apply.
\add{Thus}, we assume that %
\texttt{x} and \texttt{y} range over 0 to \mbox{\texttt{q}$-1$}, except that \texttt{y} ranges from 1 in multiplication.
\addx{It is the \addy{programmer's responsibility} to ensure
\texttt{y} and \texttt{q} are coprime integers in \texttt{x *=\(\sb{\texttt{q}}\) y}.}
\addx{In practice, it is sufficient that \addy{\texttt{q} is prime}, so that \texttt{y} and \texttt{q} are coprime.}
The subtraction\add{,} 
$x-_qy$\add{,} can be realized \add{using} an uncall to \mbox{\texttt{x +=\(\sb{\texttt{q}}\) y}}\add{,} and is \addx{written} as \texttt{x -=\(\sb{\texttt{q}}\) y}.
\addx{The variable on the left-hand side must not occur on the right-hand side of any of these operators.}
We assume \add{that these} operators \add{perform} in constant time.

The \texttt{n}th power of \texttt{b} can be stored in a zero-cleared variable\add{,} \texttt{z}, written \mbox{\texttt{z +=\(\sb{\texttt{q}}\) b\(\sp{\texttt{n}}\)}}, by initializing \texttt{z} with~\texttt{1} \addx{and} repeating \addx{\texttt{n}}-times
\mbox{$\texttt{z *=\(\sb{\texttt{q}}\) b}$}:
\begin{alltt}
  z += 1
  iterate int i = 0 to n-1
    z *=\(\sb{\texttt{q}}\) b
  end
\end{alltt}
For 
notational simplicity, \addx{we write} \texttt{z -=\(\sb{\texttt{q}}\) b\(\sp{\texttt{n}}\)}
as the inverse of
\texttt{z +=\(\sb{\texttt{q}}\) b\(\sp{\texttt{n}}\)}
to zero-clear \texttt{z}. In an 
implementation
the arguments of
modular arithmetic operators cannot be larger than the largest representable integer in \add{a} particular reversible programming language.

\paragraph{A Reversible Rabin--Karp Matcher}

\add{Figure~\ref{fig:rabin-karp} shows the program for the reversible Rabin--Karp algorithm.} The program consists of three procedures and \add{uses} procedure \texttt{match} 
in Fig.~\ref{fig:naive}.
\begin{figure}[tttt]
\begin{lstlisting}
procedure init_h(int x, i, j, X[], d, q)
  iterate int k = i to j-1 // compute hash value by Horner's rule
    x *=(*@\(\sb{\texttt{q}}\)@*) d     // shift by one radix-d digit
    x +=(*@\(\sb{\texttt{q}}\)@*) X[k]  // add the low-order radix-d digit
  end

procedure update(int t, T[], s, h, m, q)
  t -=(*@\(\sb{\texttt{q}}\)@*) T[s]*h  // remove the high-order radix-d digit(*@\label{fig:rabin-karp:update:remove}@*)
  t *=(*@\(\sb{\texttt{q}}\)@*) d       // shift by one radix-d digit(*@\label{fig:rabin-karp:update:shift}@*)
  t +=(*@\(\sb{\texttt{q}}\)@*) T[s+m]  // add the low-order radix-d digit(*@\label{fig:rabin-karp:update:add}@*)

procedure rabinkarp(int T[], P[], d, q, stack R)
  local int t=0, int p=0, int h=0
  call init_h(p, 0, m, P, d, q)(*@\label{lst:rk:call:hp}@*) // store hash of P[0..m-1] in p
  call init_h(t, 0, m, T, d, q)(*@\label{lst:rk:call:ht}@*) // store hash of T[0..m-1] in t
  h +=(*@\(\sb{\texttt{q}}\)@*) d(*@\(\sp{\texttt{m-1}}\,\)@*)                     // precompute factor h(*@\label{lst:rk:call:pow}@*)

  iterate int s = 0 to n-m      // slide over text (*@\label{lst:rk:iterate:begin}@*)
    if p = t then               // compare hash values(*@\label{lst:rk:p=t}@*)
      call match(T, P, s, R)    // match at current shift s(*@\label{lst:rk:call:match}@*)
    fi p = t
    call update(t, T, s, h, m, q)(*@\label{lst:rk:call:update}@*)    // update hash value
  end(*@\label{lst:rk:iterate:end}@*)
  
  h -=(*@\(\sb{\texttt{q}}\)@*) d(*@\(\sp{\texttt{m-1}}\;\)@*)                          // clear h(*@\label{lst:rk:uncall:pow}@*)
  uncall init_h(t, n-m, n, T, d, q)  // clear t(*@\label{lst:rk:uncall:ht}@*)
  uncall init_h(p, 0,   m, P, d, q)  // clear p(*@\label{lst:rk:uncall:hp}@*)
  delocal int t=0, int p=0, int h=0
\end{lstlisting}
    \caption{Reversible Rabin--Karp algorithm.}
    \label{fig:rabin-karp}
\end{figure}

The main procedure \texttt{rabinkarp} is called with text \texttt{T}, pattern \texttt{P}, alphabet size \texttt{d} (including \addx{the} dummy character terminating \texttt{P}),
modulus \texttt{q}, and an initially empty stack \texttt{R} as \add{the} input. When it returns, all valid shifts are stored in \texttt{R}, and 
all other 
variables \texttt{T}, \texttt{P}, \texttt{d}, and \texttt{q} \add{remain} \addx{unchanged}.
\addx{Therefore, it is an implementation of an input-preserving injectivization of the string-matching problem shown in Eq.~\ref{eq:inj}.}

The main iteration in lines \ref{lst:rk:iterate:begin}--\ref{lst:rk:iterate:end} corresponds to \add{that} in the main procedure 
of the naive string-matching algorithm\add{,} except that the hash value \texttt{p} of \texttt{P} and the hash value \texttt{t} of \texttt{T} at shift \texttt{s} are compared. \addx{Only if} a match is possible, \add{that is} \mbox{\texttt{p = t}} is true in line \ref{lst:rk:p=t}, an exact match of \texttt{P[0..m-1]} and \texttt{T[s..s+m-1]} is performed in the then-branch by calling \texttt{match} in line \ref{lst:rk:call:match}. This exact match can update \texttt{R} with a valid shift depending on the outcome of the comparison. After the conditional, \texttt{t} at shift \texttt{s} is updated to the hash value at shift \texttt{s+1} \add{through} a call to \texttt{update}. \add{Subsequently,} the iteration continues at the next shift.
  
The pre- and post\add{-}processing before and after the main iteration are performed in lines \ref{lst:rk:call:hp}--\ref{lst:rk:call:pow} and lines \ref{lst:rk:uncall:pow}--\ref{lst:rk:uncall:hp}, respectively. In 
pre\add{-}processing, the hash values \texttt{p} and \texttt{t} are initialized by the calls to \texttt{init\char`_h} as defined in Eqs.~(\ref{eq:hashp}, \ref{eq:hasht})\add{,} and \texttt{h} is precomputed \add{using} \texttt{h +=}\(\sb{\texttt{q}}\) \texttt{d}\(\sp{\texttt{m-1}}\) as defined in Eq.~(\ref{eq:h}). 
Post\add{-}processing, a typical idiom of reversible programming to zero-clear variables, uncomputes the 
values {of \texttt{h}, \texttt{t} and \texttt{p}}. \add{Notably,} the call {of \texttt{init\_h}} in line \ref{lst:rk:call:ht} and the uncall {of \texttt{init\_h}} in line \ref{lst:rk:uncall:ht} have different indices. 
The uncall uses the fact that the last value of \texttt{t} is the hash value of \texttt{T} \add{in} the \emph{last shift} \texttt{n-m} {({\cf}, 
Eq.~(\ref{eq:hashtfinal}))}, \add{whereas} the initial value of \texttt{t} is the hash value of \texttt{T} \add{in} the \emph{first shift} \texttt{0}.

Procedure \texttt{update} computes the \add{subsequent} hash value\add{,} $t_{s+1}$ from $t_s$\add{,} 
in constant time using Eq.~(\ref{eq:phismod}).
Line \ref{fig:rabin-karp:update:remove} removes the high-order radix-\texttt{d} digit from \texttt{t}, line \ref{fig:rabin-karp:update:shift} multiplies \texttt{d}, which shifts the radix-\texttt{d} number left by \add{a one-digit} position, and line \ref{fig:rabin-karp:update:add} adds the low-order radix-\texttt{d} digit. All \add{the} operations \addx{are}
modulo~\texttt{q}.

Procedure \texttt{init\char`_h} computes in an initially zero-cleared variable \texttt{x}\add{,} the hash value $\psi(X[i..j-1])$ of a substring $X[i..j-1]$\add{,} 
using Horner's rule as shown in Eq.~(\ref{eq:psi}). In addition to \texttt{x}, the indices \texttt{i} and \texttt{j},
array \texttt{X}, 
alphabet size \texttt{d}, and 
modulus \texttt{q} \add{are used} as \add{inputs}.
The hash values $p=\psi(P[0..m-1])$ and $t_0=\psi(T[0..m-1])$ are %
computed in $\Theta(m)$.

\paragraph{{Space, Time}, and Reversibilization}

Regarding the space consumption of the program, 
no extra arguments are passed to procedure \texttt{rabinkarp} to  \add{maintain} garbage values. All local variables are allocated and deallocated in the body of the procedures, and neither stacks nor arrays are allocated. 
Therefore, no additional space is required compared \add{with} the irreversible original of the program\addx{~\cite{CLRS09:intro-alg}}.

The pre- and post\add{-}processing times are bounded by the running time of \texttt{init\char`_h}, which is $\Theta(m)$. The worst-case running time of matching is $\Theta((n-m+1)m)$, which is the same as \add{that of} the irreversible Rabin--Karp algorithm. In the worst case, the iteration repeats the $\Theta(m)$ steps of the exact match $n-m+1$ times. 
The shift \texttt{s} increases monotonically in the main iteration of procedure \texttt{rabinkarp} and a limited number of elements $T[s..s+m-1]$ of the text is accessed at each iteration. Thus, 
\add{just like the irreversible original,} 
the proposed Rabin--Karp algorithm 
computes over bounded space.

\add{Notably,} no extra space is required by the two reversible \addx{string-matching} programs. 
The speedup gained by allowing deletion as a computational resource 
(at the expense of additional heat dissipation \add{owing} to %
entropy increase~\cite{DeVos:20,Krakovsky:21})
is
a constant factor of \add{approximately} two (the uncall in procedure \texttt{match} in
case of an unsuccessful comparison). \add{Thus,} if \add{reversibility were removed} and the two programs were turned back into C-like imperative programs, speed, but no space \add{would be obtained.} The reversible programs\add{,} developed 
in this study\add{,} are in contrast to what one 
obtains from \add{mechanically reversibilizing} the string-matching algorithms using Bennett's method~\cite{Bennett:73}. \add{Bennett's method} has the advantage that it can be applied to any irreversible program; \add{however,} it requires additional space proportional to the length of the computation \add{because of} the recording of a trace {(for reversible simulations with improved space efficiency\add{,} see
\eg,~\cite{PePh15})}. \add{Therefore}, the entire run of a \addx{reversibilized} string matcher, \add{including all hash calculations and mismatches,} is recorded in the computation history. Moreover, \add{because} of the uncompute phase added by Bennett's method to clean up the trace after finding all valid shifts, the reversible string-matching program produced by \add{this} \addx{reversibilization} method is no longer 
computing over bound space.

\section{Conclusion}
\label{sec:conclusion}

In this study, we have designed
and implemented the reversible versions of a naive string\add{-}matching algorithm and the Rabin--Karp algorithm.
We have shown that the hash\add{-}update function with a reasonable restriction in the reversible Rabin--Karp algorithm is injective.
The reversible versions of a naive string matching algorithm and the Rabin--Karp algorithm have the same asymptotic running time $\mathrm{O}((n-m+1)m)$ and space
usage $\mathrm{O}(n+m)$\add{,} \addx{as}
the irreversible \add{versions}. Because 
the original inputs preserved over the runs are not regarded as garbage, both reversible algorithms are clean, {\ie}\ they produce no garbage as output. The 
main iteration
monotonically increases the shift \texttt{s} from 0 to \texttt{n-m+1}. Thus, the proposed Rabin--Karp algorithm 
is a streaming algorithm.
This property 
cannot be automatically obtained by the reversibilization of Bennett~\cite{Bennett:73} and Lange--McKenzie--Tapp~\cite{LaMT00}.
It is expected that the reversible algorithms developed in this study can be a part of other algorithms, and the insights gained from constructing the reversible algorithms can be applied in future program developments. 
Verifying conventional programs 
is not always an easy task (\eg, \cite{Angelis:22}), and %
{exploring} the reversibilization and {mechanical} verification of reversible programs will be a challenge \add{in} future work.

\paragraph*{Acknowledgments}

{The authors would like to thank} Geoff Hamilton, Temesghen Kahsai, and Maurizio Proietti for their kind invitation to contribute
to the workshop HCVS/VPT at ETAPS week 2022 in Munich
{and the anonymous reviewers for the useful feedback.}
The idea for the reversible algorithms in this \add{study} was \add{brewed} in joint work with Kaira Tanizaki and Masaki Hiraku.
The second author was supported by JSPS KAKENHI Grant Number 22K11983.

\end{document}